\documentclass[english]{article}

\usepackage{amsmath}
\usepackage{amssymb}
\usepackage{color}
\usepackage{amsthm}
\usepackage{empheq}
\usepackage{graphicx}
\usepackage{comment}

\newtheorem{lemma}{Lemma}

\usepackage{appendix}
\usepackage[utf8]{inputenc}
\usepackage{hyperref}

\usepackage[T1]{fontenc}

\newcommand*\Laplace{\mathop{}\!\mathbin\bigtriangleup}
\newcommand{\dr}{{\rm d}{\bf r}}

\begin{document}

\title{Dynamical representations of constrained multicomponent nonlinear Schrödinger equations in arbitrary dimensions}

\author{M. Gulliksson* and M. Ögren**}

\maketitle

\bigskip
  \footnotesize

*(Corresponding author)  \textsc{School of Science and Technology, Örebro University, 701 82 Örebro, Sweden}\par\nopagebreak
  \textit{E-mail address:} \texttt{marten.gulliksson@oru.se}

  \medskip

  **\textsc{School of Science and Technology, Örebro University, 701 82 Örebro, Sweden, Hellenic Mediterranean University, P.O. Box 1939, GR-71004, Heraklion, Greece.}\par\nopagebreak
  \textit{E-mail address:} \texttt{magnus.ogren@oru.se}

  \medskip

\begin{abstract}
   We present new approaches for solving constrained multicomponent nonlinear Schrödinger equations in arbitrary dimensions. The idea is to introduce an artificial time and solve an extended damped second order dynamic system whose stationary solution is the solution to the time-independent nonlinear Schrödinger equation. 
Constraints are often considered by projection onto the constraint set,
here we include them explicitly into the dynamical system. 
We show the applicability and efficiency of the methods on examples of relevance in modern physics applications.
\end{abstract}

\section{Introduction}

We will describe and numerically test a family of  methods for finding stationary solutions to   coupled multicomponent nonlinear Schrödinger equations in arbitrary number of spatial dimensions using different representations of damped dynamical systems.

The nonlinear Schrödinger equation (NLSE) \cite{Fibich2015} is widely used today as a model, for example in nonlinear optics \cite{OgrenEtAlSolitons2017}; in super-conductivity, modeled with the related Ginzburg-Landau equation \cite{SorensenEtAl2017}; in models for dark matter \cite{ChavanisEtAl2011}; and the NLSEs have also been used in the description of water waves \cite{RozemannEtAl2020}, as well as for modeling the occurrence of rogue surface waves at sea \cite{SystheEtAl2008}. 

Here we present examples of constrained vector-NLSEs \cite{GeneralNsolitonFeng2014}, often called coupled Gross-Pitaevskii equations in the settings of a mean-field description of rotating bosonic atoms in different components. This is currently an active area of research in experiments on ultra-cold atomic gases \cite{BeattieEtAl2013,LannigEtAl2020} in need of tools for theoretical investigations.

We will  limit the presentation to the most common form of NLSEs, i.e. with cubic nonlinearities, but the nonlinear terms can be more general functionals of the densities of the components, such is used in modeling, for example of Tonks gases \cite{OgrenEtAlSolitary2005}, or superfluid Fermi-Bose mixtures \cite{OgrenEtAlRotational2020}, to mention just a few realizations of different nonlinear couplings. The readers can straightforwardly modify the presented constraints and dynamic equations to represent future models of interest.

There are other numerical ways to solve for  the stationary solution of the
NLSE numerically. The so called imaginary time dependent SE is common \cite{AntoineDuboscg2014}.
This means that after discretization in space (e.g., by finite differences, finite elements, or
spectral decomposition)  solve a first order damped time dependent equation numerically,
see \cite{SmyrlisEtAl2004}. Sometimes these methods are called steepest descent methods
When the NLSE is independent of time  it is,
after discretization in space, equivalent to a  finite dimensional minimization
problem, with nonlinear constraints. Such minimization problems can generally be
solved by a variety of numerical methods including gradient descent methods, (Quasi-)
Newton methods, machine learning techniques etc. \cite{NoceWrig06_Numerical_Optimization,SraEtAl2012}.

However, as demonstrated in this article, the extended second order damped dynamical systems presented have many benefits, such as ease of implementation and low computational complexity.

We begin in the next section with the necessary notation and presentation of the problem formulated as a minimization problem with constraints. We there derive the necessary conditions for a solution. In Sec. \ref{sec:DFPM} we introduce the damped dynamical system to be used to attain a stationary solution that solves our original constrained minimization  problem. This approach to obtain the solution is an extension of the Dynamical Functional Particle Method (DFPM), see \cite{Gulliksson_book_chapter_2019}, able to handle nonlinear constraints. Our version of DFPM involves finding the Lagrange parameters and we derive the linear equations determining those. We apply DFPM to the NLSE in one dimension in detail in Sec. \ref{sec:one_dim_multicomp} based on the general results in the earlier section. We show the applicability, efficiency, and stability of DFPM in numerical tests on both one- and two-dimensional NLSEs in Sec. \ref{sec:Numerical}. Finally, we summarize and discuss our results in Sec. 
\ref{sec:Conclusions}.

\section{The multicomponent nonlinear Schrödinger Equation}

In this section we present the general $n$-component Nonlinear Schrödinger  equation (NLSE) and constraints. Further, we define the corresponding total energy and Lagrange functional to be used in later sections.

Let $\mathcal{H}$ denote a Hilbert space of complex valued vector functions $\vec{\psi}: \mathbb{R}^m\longrightarrow \mathbb{C}^n$, $m, n \in \mathbb{N}$. We use the standard notation in physics for the position vector as ${\bf r}$. The space $\mathcal{H}$ is equipped with the inner product 
\begin{equation}
\left\langle \vec{\psi}\middle| \vec{\phi} \right\rangle = 
\int \vec{\psi}^{\dagger} \,  \vec{\phi}\, \dr,
\end{equation}
and the norm
\begin{equation}
\|  \vec{\psi} \|  = 
\sqrt{
\int \vec{\psi}^{\dagger} \,  \vec{\psi}\, {\rm d}{\bf r},
}
\end{equation}
where $\dagger$ is transpose conjugation. 
We  assume that the functions in $\mathcal{H}$ have sufficient smoothness  and satisfy either periodic boundary conditions or $\lim_{\|{\bf r} \| \rightarrow \infty} \vec{\psi}({\bf r})=0$. 

The total energy is defined as  the convex functional $\widetilde{E} : \, \mathcal{H} \longrightarrow \mathbb R$%_{\ge 0}$  
 with
\begin{equation}
\widetilde{E}(\vec{\psi})=
\int 
\sum_{i=1}^m \nabla_{i} \vec{\psi}^{\dagger} W \nabla_{i} \vec{\psi}+
\vec{\psi}^{\dagger} V(x) \vec{\psi} + 
\vec{\psi}^{\dagger} 
\Gamma(\vec{\psi})\vec{{\psi}} \, \dr, 
\label{eq:energy_functional}
\end{equation}
where  $W = \text{diag}(1/(2M_1), \ldots, 1/(2M_n))$ represents different masses, and
$V({\bf r}) = \text{diag}(V_1({\bf r}), \ldots, V_n({\bf r}))$ so we might have different external potentials for different components. The matrix
\begin{equation}
\Gamma(\vec{\psi})=\left[\begin{array}{cccc}
\gamma_{11}\left|\psi_{1}\right|^{2} & \gamma_{12}\psi_{1}\bar{\psi}_{2} & \dots & \gamma_{1n}\psi_{1}\bar{\psi}_{n}\\
\gamma_{12}\psi_{2}\bar{\psi}_{1} & \gamma_{22}\left|\psi_{2}\right|^{2} & \dots & \gamma_{2n}\psi_{2}\bar{\psi}_{n}\\
\vdots & \vdots & \ddots & \vdots\\
\gamma_{1n}\psi_{n}\bar{\psi}_{1} & \gamma_{2n}\psi_{n}\bar{\psi}_{2} & \dots & \gamma_{nn}\left|\psi_{n}\right|^{2}
\end{array}\right], \gamma_{ij} \in \mathbb R,%_{\ge 0}
\end{equation}
specifies the intra- and inter-component interactions, which is repulsive if $\gamma_{ij} >0$. 
Further, we assume that there are additional global  constraints $\widetilde{G}_j : \, \mathcal{H} \longrightarrow \mathbb C$ on the form 
\begin{equation}
\label{eq:genGdef}
\widetilde{G}_{j}(\vec{\psi})=
c_{j}-\int\vec{{\psi}}^{\dagger}K_{j}\vec{{\psi}}d{\bf r}=0,\ 
j = 1, \ldots, n_c , \ c_j \in \mathbb C,
\end{equation}
where $K_j$ is a  matrix, to be specified later, not depending on $\vec{\psi}$.
Further, we define the constraint set 
$$
\widetilde{\mathcal{G}} = \left\{ \vec{\psi} : 
\widetilde{G}_j(\vec{\psi}) = 0, j=1, \ldots , n_c  \right\} .
$$
Denoting the  Lagrange multipliers as  $\lambda_{j} \in \mathbb C, j=1, \ldots, n_c$,
the Lagrange functional (extended energy) can be written as
\begin{equation}
\label{eq:Idef}
\widetilde{I}(\vec{\psi},\vec{{\lambda}})=
\widetilde{E}(\vec{\psi})+
\sum_{j}\lambda_{j}\widetilde{G}_{j}(\vec{\psi}).
\end{equation}
It is well known  \cite{ZeidlerIII} that if $\vec{\psi}^*$ is a local minimizer of 
\begin{equation}
\min_{\vec{{\psi}}\in \widetilde{\mathcal{G}}} \widetilde{E}\left(\vec{{\psi}}\right),
\end{equation}
then the first order necessary condition 
\begin{equation}
\label{eq:dIdphi0}
\dfrac{\delta \widetilde{I}}{\delta\vec{{\psi}}^{\dagger}}=
\dfrac{\delta \widetilde{E}}{\delta\vec{{\psi}}^{\dagger}}+
\sum_{j}\lambda_{j}\dfrac{\delta \widetilde{G}_{j}}{\delta\vec{{\psi}}^{\dagger}}=0
\end{equation}
is satisfied at $\vec{\psi}^*$. 

If we define the projection of $\vec{\xi}\in \mathcal{H}$ onto the tangent space of the constraint set $\widetilde{\mathcal{G}}$  as
\begin{equation}
\label{eq:PDdef}
P_{\widetilde{\mathcal{G}}}^{\perp}(\vec{\xi}) = \vec{\xi}-\sum_{j=1}^{n_c}
\alpha_j (\vec{\xi})
\dfrac{\delta \widetilde{G}_{j}}{\delta\vec{{\psi}}^{\dagger}},
\end{equation}
where $\alpha_j $ are given by the solution of the linear system 
\begin{equation}
    \label{eq:Axib}
T\vec{\alpha}=\vec{d},
\end{equation}
where
\begin{equation}
\label{eq:calcPAb}
t_{ij} = 
\left\langle   
\dfrac{\delta \widetilde{G}_{i}}{\delta\vec{{\psi}}^{\dagger}}\,
\middle|
\dfrac{\delta \widetilde{G}_{j}}{\delta\vec{{\psi}}^{\dagger}}
\right\rangle, \,
d_{i} = 
\left\langle   
\dfrac{\delta \widetilde{G}_{i}}{\delta\vec{{\psi}}^{\dagger}}  \,  \,
\middle| \vec{\xi}
\right\rangle ,
\end{equation}
then  
$$
P_{\widetilde{\mathcal{G}}}^{\perp}\left(
\sum_{j}\lambda_{j}\dfrac{\delta \widetilde{G}_{j}}{\delta\vec{{\psi}}^{\dagger}}
\right) =0,
$$
and
(\ref{eq:dIdphi0}) gives
the a related necessary condition
for a minimum
\begin{equation}
\label{eq:PdEdPhi}
P_{\widetilde{\mathcal{G}}}^{\perp}\left(\dfrac{\delta \widetilde{E}}{\delta\vec{{\psi}}^{\dagger}}
\right)=0.
\end{equation}
Note that in (\ref{eq:PdEdPhi}) the Lagrange parameters are not present. This fact will be used later.

To be more precise we get from (\ref{eq:energy_functional}) and (\ref{eq:genGdef}) 
by using partial integration for the first term that
$$
\dfrac{\delta \widetilde{E}(\vec{{\psi}})}{\delta\vec{{\psi}}^{\dagger}} = 
-W\Laplace \vec{\psi} + V\vec{\psi}  +2\pi\Gamma\left(\vec{{\psi}}\right)\vec{{\psi}},\ 
\frac{\delta \widetilde{G}_j(\vec{{\psi}})}{\delta\vec{\psi}^{\dagger}} = -
K_j\vec{\psi},
$$
and we see that (\ref{eq:dIdphi0}) is equivalent to the partial differential equation
\begin{equation}
-W\Laplace \vec{\psi} + V\vec{\psi}  +2\pi\Gamma\left(\vec{{\psi}}\right)\vec{{\psi}} 
-\sum_{j}\lambda_{j}K_{j}\vec{{\psi}} =0,
\end{equation}
where we define the Laplace operator acting on a vector  to be
$$
\Laplace \vec{\psi} = 
\left[\begin{array}{c}
\Laplace \psi_1\\
\vdots\\
\Laplace \psi_n
\end{array}\right].
$$

\section{General formulation of DFPM with constraints}
\label{sec:DFPM}

The first step in the Dynamical Functional Particle Method \cite{Gulliksson_book_chapter_2019}, DFPM, is  to  introduce an artificial time $\tau$ making all the functions and functionals in the preceeding section depend on $\tau$. Thus, $\vec{\psi}({\bf r})$ become, say, $\vec{{\Psi}}({\bf r},\tau)$ and we drop all the tildes in the notation used above to indicate the new dependence on $\tau$, e.g., $\widetilde{E}(\vec{\psi})$ will become $E(\vec{\Psi})$ . Then we formulate the damped dynamical system
\begin{equation}
\label{eq:dampeddIdPsi}
\dfrac{\partial^2 \vec{\Psi}}{\partial \tau^2}+
\eta\dfrac{\partial \vec{\Psi}}{\partial \tau}=
-\dfrac{\delta I(\vec{{\Psi}},\vec{{\lambda}})}{\delta\vec{{\Psi}}^{\dagger}},
\end{equation}
with  initial conditions $\vec{\Psi}_0=\vec{{\Psi}}({\bf r},0), \vec{\Phi}_0=\frac{\partial \vec{\Psi}}{\partial \tau} ({\bf r},0)$, and where
$\eta >0$ is a damping parameter.

We  need to satisfy the constraints at the stationary point. % $\vec{\Psi}^*$. 
This can be done either by projection onto the constraint set or by formulating additional damped dynamical systems for the constraints, see below. Assume that the number of constraints satisfied by projection is $n_P$ then
$G_j(\vec{{\Psi}}(\tau)) = 0, j=1, \ldots , n_P\leq n_c$
and we have the constraint set 
$\mathcal{D}(\tau) =$ $\left\{ \vec{\Psi}(\tau) : G_j(\vec{\Psi}(\tau)) = 0,j=1, \ldots, n_P \right\}$. The functional derivatives in (\ref{eq:dIdphi0}) should be taken on the constraint set $\mathcal{D}(\tau)$, i.e., 
\begin{equation}
\label{eq:ProjdIdphi}
\dfrac{\delta I(\vec{{\Psi}},\vec{{\lambda}})}{\delta\vec{{\Psi}}^{\dagger}} \left. \right|_{\mathcal{D}} =
P_{\mathcal{D}}^{\perp}
\left(
\dfrac{\delta E}{\delta\vec{{\Psi}}^{\dagger}}+\sum_{j}\lambda_{j}\dfrac{\delta G_{j}}{\delta\vec{{\Psi}}^{\dagger}}
\right),
\end{equation}
where $P_{\mathcal{D}}^{\perp}$ is the (orthogonal) projection on the tangent space of the constraint set $\mathcal{D}(\tau)$ defined 
as in 
(\ref{eq:PDdef}). The dynamical system corresponding to (\ref{eq:dampeddIdPsi}) becomes
\begin{equation}
\dfrac{\partial^2 \vec{\Psi}}{\partial \tau^2}+
\eta\dfrac{\partial \vec{\Psi}}{\partial \tau}=-P_{\mathcal{D}}^{\perp}
\left(
\dfrac{\delta E}{\delta\vec{{\Psi}}^{\dagger}}+\sum_{j}\lambda_{j}\dfrac{\delta G_{j}}{\delta\vec{{\Psi}}^{\dagger}}
\right),
\end{equation}
or since $P_{\mathcal{D}}^{\perp}
\left(
\frac{\delta G_{j}}{\delta\vec{{\Psi}}^{\dagger}}
\right)=0, j=1, \ldots, n_P$ 
\begin{equation}
\label{eq:DFPMproj}
\dfrac{\partial^2 \vec{\Psi}}{\partial \tau^2}+
\eta\dfrac{\partial \vec{\Psi}}{\partial \tau}=
-P_{\mathcal{D}}^{\perp}
\left(
\dfrac{\delta E}{\delta\vec{{\Psi}}^{\dagger}}
\right)-
\sum_{j=n_P+1}^{n_c}\lambda_{j}
P_{\mathcal{D}}^{\perp}\left(
\dfrac{\delta G_{j}}{\delta\vec{{\Psi}}^{\dagger}}
\right).
\end{equation}
The remaining constraints are forced to be satisfied at the stationary solution by solving additional damped dynamical systems
\begin{equation}
\label{eq:Gdynm}
    \begin{array}{c}
\dfrac{\partial^2 G_{n_P+1}}
{\partial \tau^2}+
\eta\dfrac{\partial G_{n_P+1}}
{\partial \tau}=-k_{n_P+1}G_{n_P+1},\\[4mm]
\dfrac{\partial^2 G_{n_P+2}}
{\partial \tau^2}+
\eta\dfrac{\partial G_{n_P+2}}
{\partial \tau}=-k_{n_P+2}G_{n_P+2},\\
\vdots\\
\dfrac{\partial^2 G_{n_c}}
{\partial \tau^2}+
\eta\dfrac{\partial G_{n_c}}
{\partial \tau}=-k_{n_c}G_{n_c} ,
\end{array}
\end{equation}
where $k_j>0$ can be chosen in order to optimize the numerical performance of the method. Note that by definition these equations ensure that $G_j \rightarrow 0, j = n_P+1, \ldots , n_c$ when $\tau\rightarrow \infty$. 

The unknown  Lagrange parameters $\lambda_j, j= n_P+1, \ldots, n_c$ are determined by substituting the known $G_j, j = n_P+1, \ldots , n_c$ into (\ref{eq:Gdynm}) and then using (\ref{eq:DFPMproj}) as we will show in detail in Sec. \ref{sec:Lagrange}.

Finally, note that it is possible to have no projection, $n_P=0$, as well as no additional damped systems for the constraints, $n_P=n_c$. 

\subsection{Finding the Lagrange parameters}
\label{sec:Lagrange}

First assume  that the first $n_P$ constraints
are satisfied. Further, for a simpler notation, we use inner products (bra-ket) for the integrals. We will use (\ref{eq:Gdynm}) to get the equations for the Lagrange parameters  $\lambda_j, j= n_P+1, \ldots , n_c$ in (\ref{eq:DFPMproj}), i.e.,  we need the first and second order derivatives of $G_j$ w.r.t. $\tau$. 
For notational convenience we drop the brackets for the projection, i.e., $P_{\mathcal{D}}^{\perp}(\vec{\xi})=P_{\mathcal{D}}^{\perp}\vec{\xi}$.

From (\ref{eq:genGdef})
we have that
\begin{equation}
\dfrac{\partial G_j}{\partial \tau}=-\left\langle   \dfrac{\partial \vec{\Psi} }{\partial \tau}\middle| K_j\vec{{\Psi}}\right\rangle
-\left\langle \vec{\Psi}\middle|K_{j}\dfrac{\partial \vec{\Psi}}{\partial \tau} \right\rangle,
\label{eq:dGjdtau}
\end{equation}
and 
\begin{equation}
\label{eq:dddGGG}
\dfrac{\partial^2 G_j}{\partial \tau^2}=
-\left\langle   \dfrac{\partial^2 \vec{\Psi} }{\partial \tau^2}\middle| K_j\vec{{\Psi}}\right\rangle
-2\left\langle   \dfrac{\partial \vec{\Psi} }{\partial \tau}\middle| K_j\dfrac{\partial \vec{\Psi} }{\partial \tau}\right\rangle
-\left\langle \vec{{\Psi}} \middle| K_j  \dfrac{\partial^2 \vec{\Psi} }{\partial \tau^2}\right\rangle
\end{equation}
and with  (\ref{eq:dGjdtau}) and (\ref{eq:dddGGG}) we get 
\begin{equation}
\label{eq:d2Gdg}
\begin{array}{l}
\dfrac{\partial^2 G_j}{\partial \tau^2} + \eta \dfrac{\partial G_j}{\partial \tau} =
-\left\langle   \dfrac{\partial^2 \vec{\Psi} }{\partial \tau^2}\middle| K_j\vec{{\Psi}}\right\rangle
-2\left\langle   \dfrac{\partial \vec{\Psi} }{\partial \tau}\middle| K_j\dfrac{\partial \vec{\Psi} }{\partial \tau}\right\rangle
-\left\langle \vec{{\Psi}} \middle| K_j  \dfrac{\partial^2  \vec{\Psi}}{\partial \tau^2}\right\rangle\\
 -\eta\left\langle   \dfrac{\partial \vec{\Psi} }{\partial \tau}\middle| K_j\vec{{\Psi}}\right\rangle
-\eta\left\langle \vec{\Psi}\middle|K_{j}\dfrac{\partial \vec{\Psi}}{\partial \tau} \right\rangle = \\
 -\left\langle
 \dfrac{\partial^2 \vec{\Psi}}{\partial \tau^2}+
 \eta\dfrac{\partial \vec{\Psi} }{\partial \tau}\middle|
 K_{j}\vec{\Psi}\right\rangle -
2\left\langle   \dfrac{\partial \vec{\Psi} }{\partial \tau}\middle| K_j\dfrac{\partial \vec{\Psi} }{\partial \tau}\right\rangle
-\left\langle \vec{\Psi}\middle| K_j 
\left(
 \dfrac{\partial^2 \vec{\Psi}}{\partial \tau^2}+
 \eta\dfrac{\partial \vec{\Psi} }{\partial \tau}
 \right)
 \right\rangle .
\end{array}
\end{equation}
With  (\ref{eq:DFPMproj}) inserted into (\ref{eq:d2Gdg})
we get 
\begin{equation}
\label{eq:damp1}
\begin{array}{l}
\dfrac{\partial^2 G_j}{\partial \tau^2} + \eta \dfrac{\partial G_j}{\partial \tau} =
 -\left\langle  
 P_{\mathcal{D}}^{\perp}\dfrac{\delta E}{\delta\vec{{\Psi}}^{\dagger}}-\sum_{k}\lambda_{k}P_{\mathcal{D}}^{\perp}K_{k}\vec{{\Psi}}
 \middle|
 K_{j}\vec{{\Psi}}
 \right\rangle -
2\left\langle  
\dfrac{\partial \vec{\Psi}}{\partial \tau} \middle| K_{j}\dfrac{\partial \vec{\Psi}}{\partial \tau}
\right\rangle\\
-\left\langle \vec{{\Psi}} \middle|
K_{j} ( P_{\mathcal{D}}^{\perp}\dfrac{\delta E}{\delta\vec{{\Psi}}^{\dagger}}-\sum_{k}\lambda_{k}
P_{\mathcal{D}}^{\perp}K_{k}\vec{{\Psi}} )\right\rangle .
\end{array}
\end{equation}
We want to solve for the Lagrange parameters so we expand the bra-ket's in (\ref{eq:damp1})
\begin{equation}
\begin{array}{l}
\dfrac{\partial^2 G_j}{\partial \tau^2} + \eta \dfrac{\partial G_j}{\partial \tau} =
- \left\langle
 P_{\mathcal{D}}^{\perp}\dfrac{\delta E}{\delta\vec{{\Psi}}^{\dagger}} 
 |
 K_{j}\vec{{\Psi}}
  \right\rangle
  +\sum_{k}\lambda_{k}
\left\langle
P_{\mathcal{D}}^{\perp} K_{k}\vec{{\Psi}} \middle|  K_{j}\vec{{\Psi}}
 \right\rangle\\
-2\left\langle  
\dfrac{\partial \vec{\Psi}}{\partial \tau} \middle| K_{j}\dfrac{\partial \vec{\Psi}}{\partial \tau}
\right\rangle
-\left\langle 
\vec{{\Psi}} \middle|
K_{j} P_{\mathcal{D}}^{\perp}\dfrac{\delta E}{\delta\vec{{\Psi}}^{\dagger}}
\right\rangle
+\sum_{k}\lambda_{k}\left\langle \vec{{\Psi}} \middle|
K_{j} P_{\mathcal{D}}^{\perp}K_{k}\vec{{\Psi}} )\right\rangle ,
\end{array}
\end{equation}
and compare with the  right hand side of (\ref{eq:Gdynm}).
By defining the matrix  $A\in \mathbb{C}^{(n_c-n_P) \times (n_c-n_P)}$ and the vector
$\vec{b}\in \mathbb{C}^{n_c-n_P}$ with elements 
\begin{align}
& a_{j-n_P,k-n_P}=
\left\langle
 P_{\mathcal{D}}^{\perp}K_{k}\vec{{\Psi}} \middle|  K_{j}\vec{{\Psi}}
 \right\rangle + \left\langle \vec{{\Psi}} \middle|
K_{j} P_{\mathcal{D}}^{\perp}K_{k}\vec{{\Psi}} )\right\rangle, \label{eq:genA} \\
& b_{j-n_P} =k_j \left( c_j - \left\langle \vec{{\Psi}} \middle|
K_{j}\vec{{\Psi}} )\right\rangle \right) + \nonumber \\
& + \left\langle
 P_{\mathcal{D}}^{\perp}\dfrac{\delta E}{\delta\vec{{\Psi}}^{\dagger}} 
 \middle|
 K_{j}\vec{{\Psi}}
  \right\rangle +
  2\left\langle  
\dfrac{\partial \vec{\Psi}}{\partial \tau} \middle| K_{j}\dfrac{\partial \vec{\Psi}}{\partial \tau}
\right\rangle +
\left\langle 
\vec{{\Psi}} \middle|
K_{j} P_{\mathcal{D}}^{\perp}\dfrac{\delta E}{\delta\vec{{\Psi}}^{\dagger}}
\right\rangle,
\label{eq:genb}
\end{align}
where $j,k=n_P+1,\ldots, n_c$
the Lagrange parameters 
$\vec{\lambda}\in \mathbb{C}^{n_c-n_P}$
are found by solving the linear system
\begin{equation}
A\vec{\lambda} = \vec{b}.
\label{eq:Alambdab}
\end{equation}

The results in (\ref{eq:genA}) and  (\ref{eq:genb}) are the cornerstone of our approach, valid for arbitrary number of components and dimensions. In the sections to come we apply it to specific problems and evaluate the resulting numerical performance. However, we have to restrict ourselves to  a limited number of test cases in one and two dimensions. 

\section{Normalization and rotational constraints for $n$ components in one dimension}
\label{sec:one_dim_multicomp}

Here we consider  one spatial dimension with a constraint on the angular momentum and on the interval $[-\pi,\pi]$ with 
periodic boundary conditions $\vec{\Psi}(-\pi) = \vec{\Psi}(\pi)$, i.e., a ring geometry. From (\ref{eq:energy_functional}) we have the energy
\begin{equation}
E\left(\vec{{\Psi}}\right)=
\int_{-\pi}^{\pi}
\dfrac{\partial \vec{\Psi}^{\dagger}}{\partial x}W
\dfrac{\partial \vec{\Psi}}{\partial x} + 
\vec{\Psi}^{\dagger} V \vec{\Psi}+
\pi\vec{\Psi}^{\dagger}\Gamma\left(\vec{{\Psi}}\right)
\vec{{\Psi}} \, {\rm d}x,\label{eq:energy_functional2}
\end{equation}
and
\begin{equation}
\frac{\delta E(\vec{{\Psi}})}{\delta\vec{{\Psi}}^{\dagger}}=
-W\frac{\partial^{2}\vec{{\Psi}}}{\partial x^{2}} + V\vec{\Psi} +2\pi\Gamma\left(\vec{{\Psi}}\right)\vec{{\Psi}}.
\label{eq:Euler_Lagrange_to_energy_functional2}
\end{equation}
The constraints are 
\begin{equation}
G_{j}=c_{j}-\int_{-\pi}^{\pi}
\vec{{\Psi}}^{\dagger}K_{j}\vec{{\Psi}}{\rm d}x=0, 
j = 1 , \ldots, n+1,
\label{eq:App_A_constraints_on_component_form}
\end{equation}
where
\begin{equation}
\begin{array}{l}
\:K_{1}=\left[\begin{array}{cccc}
1 & 0 & \dots & 0\\
0 & 0 &  & \vdots\\
\vdots &  & \ddots & \vdots\\
0 & \dots & \dots & 0
\end{array}\right],\ldots ,
\:K_{n}=\left[\begin{array}{cccc}
0 & 0 & \dots & 0\\
0 & 0 &  & \vdots\\
\vdots &  & \ddots & \vdots\\
0 & \dots & \dots & 1
\end{array}\right],\\
\vspace{3mm}
K_{n+1}=-i\left[\begin{array}{cccc}
\frac{\partial}{\partial x} & 0 & \dots & 0\\
0 & \frac{\partial}{\partial x} &  & \vdots\\
\vdots &  & \ddots & 0\\
0 & \dots & 0 & \frac{\partial}{\partial x}
\end{array}\right].
\end{array}
\label{eq:Kdef}
\end{equation}
In our general context we have here $n_c=n+1$ and the $K_{j}$ in addition fulfills $K_{n+1}=L\sum_{k<n+1}K_{k}$ for
the (angular-) momentum operator 
\begin{equation}
L=-i\frac{\partial}{\partial x}. \label{eq:definition_angular_momentum_operator}
\end{equation}
For these constraints we have 
$$
\dfrac{\delta G_{j}}{\delta\vec{{\Psi}}^{\dagger}} = -K_j\vec{\Psi}_j, \, j=1, \ldots , n+1.
$$

There are some special cases that are particularly interesting.

We may project all the constraints giving no Lagrange parameters  to determine. However, this requires the constraints to be satisfied, which for nonlinear constraints usually is a non-trivial computational task. For details, see the following subsection.

There are two other special cases  that for our problem setting are more  interesting from an algorithmic point of view and  lends themselves to an efficient algorithm and numerical tests. 

Firstly, it is natural to consider the case of projecting only some or all of the normalization constraints, $1\leq n_P\leq n$, since the projection in practice is just a trivial normalization.

Secondly, we may consider only dynamically damped constraints, $n_P=0$, that will not force any of the constraints to be satisfied except at the stationary point, which we have also experienced can lead to less sensitivity on the initial conditions.

\subsection{Only projection no dynamically damped constraints}
\label{sec:OnlyProj}

Let us first look at the case 
$n_P=n+1$, i.e., we project all the constraints. Here,  (\ref{eq:DFPMproj})  becomes
\begin{equation}
\label{eq:OnlyProjDyn}
\dfrac{\partial^2 \vec{\Psi}}{\partial \tau^2}+
\eta\dfrac{\partial \vec{\Psi}}{\partial \tau}=
-P_{\mathcal{D}}^{\perp}
\left(
\dfrac{\delta E}{\delta\vec{{\Psi}}^{\dagger}}
\right).
\end{equation}
In this case we do not have any unknown Lagrange parameters and we need only to determine the form of the projection $P_{\mathcal{D}}^{\perp}$ in (\ref{eq:PDdef}). This we get from 
 (\ref{eq:calcPAb}), i.e., 
$$
t_{ij} = 
\left\langle
K_i\vec{\Psi} \middle|  
K_j\vec{\Psi}
 \right\rangle , \,
d_{i} = 
-\left\langle 
K_i\vec{\Psi}
 \,  \,
\middle| \vec{\xi}
\right\rangle ,
$$
or in more detail
\begin{align}
\label{eq:Td2}
T =
\left[
\begin{array}{cccc}
 c_{1}   &  &  & \left\langle
 {\Psi}_{1} \middle| L{\Psi}_{n} 
 \right\rangle\\
   &   \ddots  &  & \vdots\\
    &   & c_{n} & \left\langle
 {\Psi}_{n}  \middle| L{\Psi}_{n} 
 \right\rangle\\[2mm]
 \left\langle
 {\Psi}_{1} \middle| L{\Psi}_{n} 
 \right\rangle & \cdots &  \left\langle
 {\Psi}_{n}  \middle| L{\Psi}_{n}  \right\rangle &
 \sum_{l=1}^{n} 
 \left\langle
\Psi_l  \middle|   L^2 \Psi_l
\right\rangle
\end{array}
\right], \\
\vec{d}=-
\left[
\begin{array}{c}
\left\langle 
\Psi_1 \middle|  \xi_1
\right\rangle \\
\vdots \\
\left\langle 
  \Psi_n \middle| \xi_n
\right\rangle \\[2mm]
\sum_{l=1}^{n} 
\left\langle L \Psi_l \middle|   \xi_l 
\right\rangle
\end{array}
\right].
\nonumber
\end{align}
The linear system given by the matrix $T$ and vector $\vec{d}$ in (\ref{eq:Td2}) can be efficiently solved, with order  $n$ operations, using  
sparse Gaussian elimination.

However, the constraints have to be satisfied. This is trivial in the case of normalization constraints since any approximate numerical solution can be normalized without almost any cost. The constraint on the angular momentum is more complicated to satisfy since it's  nonlinear. However, we believe the idea of using only projection is principally interesting and thus we describe a possible formulation. Assume that we have $\vec{\Psi}(x,\tau)$ from the solution of our dynamical system (\ref{eq:OnlyProjDyn}) that is not generally satisfying the constraints. In order to do so we can solve the least norm minimization problem
\begin{equation}
    \min_{\vec{\xi} \in \mathcal{D}(\tau)} 
    \| \vec{\xi} - \vec{\Psi} \|,
\end{equation}
or in more detail 
\begin{align}
  &  \min_{\vec{\xi}} \| \vec{\xi} - \vec{\Psi} \|  \\
 &  \text{subject to }\ %\nonumber\\ 
% & \hspace{1.5cm}  
 c_i - \left\langle
   {\xi}_{i} \middle| {\xi}_{i} \right\rangle =0, i=1, \ldots , n,\\
&  \hspace{1.5cm}  c_{n+1} - \sum_{i=1}^{n} 
\left\langle {\xi}_{i} \middle| L{\xi}_{i}\right\rangle =0.
\end{align}
Finding the optimal $\vec{\xi} (x)$ above is certainly possible but requires an iterative process that will presumably be much more costly in a numerical implementation than the other two approaches described below.  Therefore, we do not expect this to be an efficient algorithm and will not include this approach in the numerical tests.

\subsection{Projection of one or more normalization constraints}
\label{subsec_proj_1_or_more_norm}
 We start by consider the projection of the first  $n_P$  normalization constraints. Then the projection is explicitly given as
$$
 P_{\mathcal{D}}^{\perp}(\vec{\xi})  = 
\vec{\xi}-
\left[
\Psi_1   \langle \Psi_1 \middle|  \xi_1\rangle/c_1 , \Psi_2   \langle \Psi_2 \middle| \xi_2\rangle/c_2, \ldots,  \Psi_{n_P}  \langle \Psi_{n_P} \middle| \xi_{n_P}\rangle/c_{n_P},0, \ldots, 0
\right]^T
$$
and
\begin{equation}
\label{eq:ddtauproj2}
\dfrac{\partial^2 \vec{\Psi}}{\partial \tau^2}+
\eta\dfrac{\partial \vec{\Psi}}{\partial \tau}=-
P_{\mathcal{D}}^{\perp}\left(
\dfrac{\delta E}{\delta\vec{{\Psi}}^{\dagger}}
\right)-
\sum_{j=n_P+1}^{n+1}\lambda_{j}
P_{\mathcal{D}}^{\perp}\left( K_{j}\vec{\Psi} \right).
%\left(\dfrac{\delta G_{n+1}}{\delta\vec{{\Psi}}^{\dagger}}\right).
\end{equation}
We can simplify (\ref{eq:ddtauproj2}) slightly by noting that
$P_{\mathcal{D}}^{\perp}\left( K_{j}\vec{\Psi} \right) = \Psi_j \vec{e}_j, n_P+1 \leq j \leq n$ giving
\begin{equation}
\dfrac{\partial^2 \vec{\Psi}}{\partial \tau^2}+
\eta\dfrac{\partial \vec{\Psi}}{\partial \tau}=-
P_{\mathcal{D}}^{\perp}\left(
\dfrac{\delta E}{\delta\vec{{\Psi}}^{\dagger}}
\right)-
\sum_{j=n_P+1}^{n}\lambda_{j} \Psi_j \vec{e}_j -
\lambda_{n+1}
P_{\mathcal{D}}^{\perp}\left( K_{n+1}\vec{\Psi} \right).
%\left(\dfrac{\delta G_{n+1}}{\delta\vec{{\Psi}}^{\dagger}}\right).
\end{equation}
In the following lemma we give the details necessary to efficiently calculate the Lagrange parameters. For notational convenience we drop the brackets for the projection, i.e., $P_{\mathcal{D}}^{\perp}(\vec{\xi})=P_{\mathcal{D}}^{\perp}\vec{\xi}$.
\begin{lemma}
In (\ref{eq:genA}) and
(\ref{eq:genb}) we have
\[
\left\langle
P_{\mathcal{D}}^{\perp} K_{k}\vec{{\Psi}} \middle|  K_{j}\vec{{\Psi}}
 \right\rangle = 
 \left\langle \vec{{\Psi}} \middle|
K_{j} P_{\mathcal{D}}^{\perp}K_{k}\vec{{\Psi}} \right\rangle =
 \]

\begin{empheq}[left={\empheqlbrace}]{alignat=2}
& \delta_{kj}c_j,
\  n_P+1\leq j,k \leq n, \label{eq:deltajk}\\
& \left\langle {\Psi}_j \middle| L {\Psi}_n \right\rangle , \ k=n+1, \ j= n_P+1, \ldots n  \label{eq:phijLphin}\\
&  \sum_{l=1}^{n} 
 \left\langle
\Psi_l  \middle|   L^2 \Psi_l
\right\rangle
- \sum_{l=1}^{m} 
\left\langle  \Psi_l \middle| L \Psi_l
\right\rangle^2/c_l , \ j=k=n+1 , \label{eq:phijL2phij}
\end{empheq} %% PROBLEMS HERE ???

and
\begin{equation}
\label{eq:symrhs}
\left\langle 
 P_{\mathcal{D}}^{\perp}\frac{\delta E}{\delta\vec{{\Psi}}^{\dagger}}\middle|
K_j\vec{{\Psi}}
\right\rangle =
\left\langle 
\vec{{\Psi}} \middle|
K_j P_{\mathcal{D}}^{\perp}\frac{\delta E}{\delta\vec{{\Psi}}^{\dagger}}
\right\rangle = 
\end{equation}
\begin{empheq}
[left={\empheqlbrace}]{alignat=2}
& \left\langle \Psi_{j}\middle|\frac{\delta E}{\delta\overline{\Psi}_j}\right\rangle  , j=n_P+1, \ldots, n, 
\label{eq:phijdEdphij}\\
& \sum_{l=1}^{n} 
 \left\langle
\frac{\delta E}{\delta\overline{\Psi}_l}   \middle|   L \Psi_l
\right\rangle -
\sum_{l=1}^{n_P} 
 \left\langle
 \Psi_l   
\middle| L\Psi_l 
 \right\rangle 
 \left\langle
 \frac{\delta E}{\delta\overline{\Psi}_l}
  \middle|
  \Psi_l 
\right\rangle/c_l , j=n+1. \label{eq:sumphilLdEdphil}
 \end{empheq}

 \end{lemma}
\begin{proof}
First consider (\ref{eq:deltajk}) and note that 
$K_{k}\vec{{\Psi}} = \Psi_k \vec{e}_k$ and thus $P_{\mathcal{D}}^{\perp} K_{k}\vec{{\Psi}} =  \Psi_k \vec{e}_k$ giving $\left\langle
P_{\mathcal{D}}^{\perp} K_{k}\vec{{\Psi}} \middle|  K_{j}\vec{{\Psi}}
 \right\rangle = \left\langle \Psi_k \vec{e}_k \middle|  \Psi_j \vec{e}_j \right\rangle = \delta_{kj}c_j$. A similar argument proves that $\left\langle \vec{{\Psi}} \middle|
K_{j} P_{\mathcal{D}}^{\perp}K_{k}\vec{{\Psi}} )\right\rangle= \delta_{kj}c_j$.

Turning to  (\ref{eq:phijLphin}) we first note that $\left\langle
P_{\mathcal{D}}^{\perp} K_{n+1}\vec{{\Psi}} \middle|  K_{j}\vec{{\Psi}}
 \right\rangle = \left\langle
P_{\mathcal{D}}^{\perp} K_{n+1}\vec{{\Psi}} \middle|  \Psi_j \vec{e}_j
 \right\rangle  = \left\langle
L{\Psi}_j \middle|  \Psi_j
 \right\rangle$ and similarly 
 $\left\langle \vec{{\Psi}} \middle|
K_{j} P_{\mathcal{D}}^{\perp}K_{k}\vec{{\Psi}} \right\rangle = 
\left\langle \vec{{\Psi}} \middle|
K_{j} P_{\mathcal{D}}^{\perp}K_{n+1}\vec{{\Psi}} \right\rangle= \left\langle \vec{{\Psi}} \middle|
L\Psi_j \vec{e}_j \right\rangle =  \left\langle \Psi_j   \middle|
L\Psi_j \right\rangle$. By partial integration we get   $\left\langle
L{\Psi}_j \middle|  \Psi_j \right\rangle = \left\langle \Psi_j   \middle|
L\Psi_j \right\rangle$.

To prove (\ref{eq:phijL2phij}) we look in detail on 
\begin{align}
\left\langle
P_{\mathcal{D}}^{\perp} K_{n+1}\vec{{\Psi}} \middle|  K_{n+1}\vec{{\Psi}}
 \right\rangle =\\
 \sum_{l=1}^{n_P} 
 \left\langle
(L \Psi_l- \Psi_l  \langle \Psi_1 \middle|   L \Psi_l
\rangle/c_l
\middle| L\Psi_l 
 \right\rangle + 
 \sum_{l=n_P+1}^{n} 
 \left\langle
L \Psi_l  \middle|   L \Psi_l
\right\rangle
=\\
  \sum_{l=1}^{n} 
 \left\langle
L \Psi_l  \middle|   L \Psi_l
\right\rangle
- \sum_{l=1}^{n_P} 
\left\langle  \Psi_l \middle| L \Psi_l
\right\rangle^2/c_l,
\end{align} 
and
\begin{align}
 \left\langle \vec{{\Psi}} \middle|
K_{n+1} P_{\mathcal{D}}^{\perp}K_{n+1}\vec{{\Psi}} \right\rangle =\\
 \sum_{l=1}^{n_P} 
 \left\langle
 \Psi_l \middle| 
L(L \Psi_l- \Psi_l  
\left\langle \Psi_1 \middle|   L \Psi_l
\right\rangle/c_l)
\right\rangle+ 
 \sum_{l=n_P+1}^{n} 
 \left\langle
 \Psi_l  \middle|   L^2 \Psi_l
\right\rangle
=\\
  \sum_{l=1}^{n} 
 \left\langle
\Psi_l  \middle|   L^2 \Psi_l
\right\rangle
- \sum_{l=1}^{n_P} 
\left\langle  \Psi_l \middle| L \Psi_l
\right\rangle^2/c_l .
\end{align} 
Using partial integration we have $\left\langle
 L\Psi_l \middle|  L\Psi_l
 \right\rangle =   \left\langle
 \Psi_l \middle|  L^2\Psi_l
 \right\rangle$ showing the result.
 
 In (\ref{eq:phijdEdphij}) we again notice that $K_{j}\vec{{\Psi}} = \Psi_j \vec{e}_j$ and $P_{\mathcal{D}}^{{\perp}^{\dagger}}K_{j}\vec{{\Psi}} = \Psi_j \vec{e}_j$ so it remains to prove  that $\left\langle \Psi_{j}\middle|\frac{\delta E}{\delta\overline{\Psi}_j}\right\rangle = \left\langle \frac{\delta E}{\delta\overline{\Psi}_j} \middle| \Psi_{j} \right\rangle$ which can be done by integration by parts.
 
Finally we prove  (\ref{eq:sumphilLdEdphil}) by first looking at
\begin{align}
\left\langle 
 P_{\mathcal{D}}^{\perp}\frac{\delta E}{\delta\vec{{\Psi}}^{\dagger}}\middle|
K_{n+1}\vec{{\Psi}} 
\right\rangle=
 \sum_{l=1}^{n_P} 
 \left\langle
\frac{\delta E}{\delta\overline{\Psi}_l} - \Psi_l  \left\langle \Psi_l \middle|  \frac{\delta E}{\delta \overline{\Psi}_l}
\right\rangle/c_l
\middle| L\Psi_l 
 \right\rangle + 
 \sum_{l=n_P+1}^{n} 
 \left\langle
\frac{\delta E}{\delta\overline{\Psi}_l}   \middle|   L \Psi_l
\right\rangle=\\
\sum_{l=1}^{n} 
 \left\langle
\frac{\delta E}{\delta\overline{\Psi}_l}  \middle|   L \Psi_l
\right\rangle
-
\sum_{l=1}^{n_P} 
 \left\langle
 \Psi_l   
\middle| L\Psi_l 
 \right\rangle 
 \left\langle
  \Psi_l 
  \middle|
  \frac{\delta E}{\delta\overline{\Psi}_l}/c_l
\right\rangle,
\end{align}
and then
\begin{align}
\left\langle 
\vec{{\Psi}} \middle|
K_{n+1} P_{\mathcal{D}}^{\perp}\frac{\delta E}{\delta\vec{{\Psi}}^{\dagger}}
\right\rangle=
 \sum_{l=1}^{m} 
 \left\langle
 \Psi_l \middle| L
\frac{\delta E}{\delta\overline{\Psi}_l}- L\Psi_l  
\left\langle \Psi_l \middle|  \frac{\delta E}{\delta\overline{\Psi}_l}
\right\rangle
 \right\rangle + 
 \sum_{l=n_P+1}^{n} 
 \left\langle
 \Psi_l \middle| L
\frac{\delta E}{\delta\overline{\Psi}_l} 
\right\rangle=\\
\sum_{l=1}^{n} 
 \left\langle
 \Psi_l \middle| L
\frac{\delta E}{\delta\overline{\Psi}_l}  
\right\rangle-
 \sum_{l=1}^{n_P} 
 \left\langle
 \Psi_l \middle| 
\frac{\delta E}{\delta\overline{\Psi}_l}
\right\rangle
\left\langle
\Psi_l \middle|
L\Psi_l  
 \right\rangle/c_l.
\end{align}
By partial integration we get (\ref{eq:sumphilLdEdphil}).
\end{proof}
Using the lemma above it is easy to see that 
in  (\ref{eq:genA}) and (\ref{eq:genb}) we have
\begin{align}
A =
2
\left[
\begin{array}{cccc}
  \left\langle
 {\Psi}_{n_P+1} \middle| {\Psi}_{n_P+1} 
 \right\rangle   &  &  & \left\langle
 {\Psi}_{n_P+1} \middle| L{\Psi}_{n} 
 \right\rangle\\
   &   \ddots  &  & \vdots\\
    &   & \left\langle
 {\Psi}_{n} \middle| {\Psi}_{n} 
 \right\rangle & \left\langle
 {\Psi}_{n}  \middle| L{\Psi}_{n} 
 \right\rangle\\
 \left\langle
 {\Psi}_{n_P+1} \middle| L{\Psi}_{n} 
 \right\rangle & \cdots &  \left\langle
 {\Psi}_{n}  \middle| L{\Psi}_{n}  \right\rangle &
 \sum_{l=1}^{n} 
 \left\langle
\Psi_l  \middle|   L^2 \Psi_l
\right\rangle
- \sum_{l=1}^{n_P} 
\left\langle  \Psi_l \middle| L \Psi_l
\right\rangle^2/c_l
\end{array}\right], \label{eq:specA} \\[5mm]
\vec{b}=
2
\left[
\begin{array}{c}
k_{n_P+1}G_{n_P+1}/2+
\left\langle 
\Psi_{n_P+1}\middle|\dfrac{\delta E}{\delta\overline{\Psi}_{1}}\right\rangle -
\left\langle  
\dfrac{\partial \Psi_1}{\partial \tau} \middle| \dfrac{\partial \Psi_1}{\partial \tau}
\right\rangle \\
\vdots \\
k_{n}G_{n}/2+
\left\langle 
\Psi_{n}\middle|\dfrac{\delta E}{\delta\overline{\Psi}_n}\right\rangle -\left\langle  
\dfrac{\partial \Psi_n}{\partial \tau} \middle| \dfrac{\partial \Psi_n}{\partial \tau}
\right\rangle \\
k_{n+1}G_{n+1}/2+ \sum_{l=1}^{n} 
 \left\langle
\dfrac{\delta E}{\delta\overline{\Psi}_l}   \middle|   L \Psi_l
\right\rangle-
 \left\langle  
\dfrac{\partial \Psi_l}{\partial \tau} \middle| L\dfrac{\partial \Psi_l}{\partial \tau}
\right\rangle
-
\sum_{l=1}^{n_P} 
 \left\langle
 \Psi_l   
\middle| L\Psi_l 
 \right\rangle 
 \left\langle
 \dfrac{\delta E}{\delta\overline{\Psi}_l}  
  \middle|
  \Psi_l 
\right\rangle/c_l
\end{array}
\right].
\label{eq:specb}
\end{align}

The system (\ref{eq:Alambdab}) with the matrix in (\ref{eq:specA}) can be solved efficiently using sparse Gaussian elimination with the number of operations of order $n-n_p$.

\subsection{Projection of all normalization constraints}
\label{sec:ProjAllNorm}

 Let us consider  $n_P=n$, i.e., we project all normalization constraints but not the last constraint on the angular momentum. This will give only one unknown Lagrange parameter and in practice a trivial projection of the normalization constraints $\Psi_j=\sqrt{c_j}\Psi_j/\|\Psi_j\|$. 
The projection is explicitly given as
\begin{equation}
 P_{\mathcal{D}}^{\perp}(\vec{\xi})  = 
\vec{\xi}-
\left[
\Psi_1   \langle \Psi_1 \middle|  \xi_1\rangle/c_1 , \Psi_2   \langle \Psi_2 \middle| \xi_2\rangle/c_2, \ldots,  \Psi_n  \langle \Psi_n \middle| \xi_n\rangle/c_n
\right]^T,
\end{equation}
and
\begin{equation}
\dfrac{\partial^2 \vec{\Psi}}{\partial \tau^2}+
\eta\dfrac{\partial \vec{\Psi}}{\partial \tau}=-
P_{\mathcal{D}}^{\perp}\left(
\dfrac{\delta E}{\delta\vec{{\Psi}}^{\dagger}}
\right)-\lambda_{n+1}
P_{\mathcal{D}}^{\perp}\left( K_{n+1}\vec{\Psi} \right).
\end{equation}
The Lagrange parameter is 
\begin{equation} \label{eq:lambda_ell_for_projected_norms}
\lambda_{n+1}=
\dfrac{
k_{n+1}G_{n+1}/2+ \sum_{l=1}^{n} 
 \left\langle
\dfrac{\delta E}{\delta\overline{\Psi}_l}   \middle|   L \Psi_l
\right\rangle
-
 \left\langle
 \Psi_l   
\middle| L\Psi_l 
 \right\rangle 
 \left\langle
 \dfrac{\delta E}{\delta\overline{\Psi}_l}  
  \middle|
  \Psi_l 
\right\rangle  /c_l
-
 \left\langle  
\dfrac{\partial \Psi_l}{\partial \tau} \middle| L\dfrac{\partial \Psi_l}{\partial \tau}
\right\rangle
}
{
\sum_{l=1}^{n} 
 \left\langle
\Psi_l  \middle|   L^2 \Psi_l
\right\rangle
- 
\left\langle  \Psi_l \middle| L \Psi_l
\right\rangle^2/c_l
}.
\end{equation}

\subsection{Only dynamically damped constraints}
\label{sec:OnlyDamp}

No projection means that $P_{\mathcal{D}}^{\perp} = I $ in (\ref{eq:DFPMproj}), $n_P=0$, and the number of unknown Lagrange parameters is $n+1$ and thus 
\begin{equation}
\dfrac{\partial^2 \vec{\Psi}}{\partial \tau^2}+
\eta\dfrac{\partial \vec{\Psi}}{\partial \tau}=-
\dfrac{\delta E}{\delta\vec{{\Psi}}^{\dagger}}-\sum_{j=1}^{n+1}\lambda_{j}\dfrac{\delta G_{j}}{\delta\vec{{\Psi}}^{\dagger}},
\end{equation}
where
$$
\dfrac{\delta G_{j}}{\delta\vec{{\Psi}}^{\dagger}} = -\Psi_j\vec{e}_j, j=1, \ldots, n, 
\dfrac{\delta G_{n+1}}{\delta\vec{{\Psi}}^{\dagger}} = -\sum_{j=1}^{n}L\Psi_j\vec{e}_j.
$$
The Lagrange parameters are found by solving the $(n+1)\times (n+1)$ linear system $A\vec{\lambda}= \vec{b}$ where we from (\ref{eq:specA}) and (\ref{eq:specb}) have
\begin{align}
\label{eq:AllDynA}
A =
\left[
\begin{array}{cccc}
\left\langle
 {\Psi}_{1} \middle| {\Psi}_{1} 
 \right\rangle  &  &  & \left\langle
 {\Psi}_{1} \middle| {\Psi}_{n} 
 \right\rangle\\
   &   \ddots  &  & \vdots\\
    &   & \left\langle
 {\Psi}_{n} \middle| {\Psi}_{n} 
 \right\rangle & \left\langle
 {\Psi}_{n}  \middle| L{\Psi}_{n} 
 \right\rangle\\
 \left\langle
 {\Psi}_{1} \middle| L{\Psi}_{n} 
 \right\rangle & \cdots &  \left\langle
 {\Psi}_{n}  \middle| L{\Psi}_{n}  \right\rangle &
 \sum_{l=1}^{n} 
 \left\langle
\Psi_l  \middle|   L^2 \Psi_l
\right\rangle
\end{array}\right],  \\[5mm]
\vec{b}=
\left[
\begin{array}{c}
k_{1}G_{1}/2+
\left\langle 
\Psi_{1}\middle|\dfrac{\delta E}{\delta\overline{\Psi}_1}\right\rangle -
\left\langle  
\dfrac{\partial \Psi_1}{\partial \tau} \middle| \dfrac{\partial \Psi_1}{\partial \tau}
\right\rangle \\
\vdots \\
k_{n}G_{n}/2+
\left\langle 
\Psi_{n}\middle|\dfrac{\delta E}{\delta\overline{\Psi}_n}\right\rangle -\left\langle  
\dfrac{\partial \Psi_n}{\partial \tau} \middle| \dfrac{\partial \Psi_n}{\partial \tau}
\right\rangle \\[3mm]
k_{n+1}G_{n+1}/2+ \sum_{l=1}^{n} 
 \left\langle
\dfrac{\delta E}{\delta\overline{\Psi}_l}   \middle|   L \Psi_l
\right\rangle-
 \left\langle  
\dfrac{\partial \Psi_l}{\partial \tau} \middle| L\dfrac{\partial \Psi_l}{\partial \tau}
\right\rangle
\end{array}
\right].
\label{eq:AllDynb}
\end{align}
For $n=1$  we have
\begin{align}
\label{eq:B_10_in_EJP_2020_upper}
A =
\left[
\begin{array}{cc}
\left\langle
 {\Psi}_{1} \middle| {\Psi}_{1} 
 \right\rangle   &   \left\langle
 {\Psi}_{1} \middle| L{\Psi}_{1} 
 \right\rangle\\
 \left\langle
 {\Psi}_{1} \middle| L{\Psi}_{1} 
 \right\rangle & 
 \left\langle
\Psi_1  \middle|   L^2 \Psi_1
\right\rangle
\end{array}\right],  \\[5mm]
\vec{b}=
\left[
\begin{array}{c}
k_{1}G_{1}/2+
\left\langle 
\Psi_{1}\middle|\dfrac{\delta E}{\delta\overline{\Psi}_1}\right\rangle -
\left\langle  
\dfrac{\partial \Psi_1}{\partial \tau} \middle| \dfrac{\partial \Psi_1}{\partial \tau}
\right\rangle \\
k_{2}G_{2}/2+
 \left\langle
\dfrac{\delta E}{\delta\overline{\Psi}_1}   \middle|   L \Psi_1
\right\rangle-
 \left\langle  
\dfrac{\partial \Psi_1}{\partial \tau} \middle| L\dfrac{\partial \Psi_1}{\partial \tau}
\right\rangle
\end{array}
\right],
\label{eq:B_10_in_EJP_2020_lower}
\end{align}
which agree with (B.10) in \cite{OgrenGulliksson2020}.

\section{Numerical tests}
\label{sec:Numerical}

In this section we present numerical calculations for some representative cases in 1D and 2D.

We want to solve (\ref{eq:DFPMproj}) through an iterative method in artificial time and start by writing (\ref{eq:DFPMproj}) as the first order system \\ \\
\begin{equation}
\begin{aligned}
&\frac{\partial  \vec{\Psi}}{\partial \tau} =  \vec{\Phi},\\
&\frac{\partial\vec{\Phi}}{\partial \tau} = -\eta\vec{\Phi}+ \vec{F}(\vec{\Psi}), %%\label{(5.10)},
\end{aligned}
\label{eq:num_test_sys}
\end{equation}
where
\begin{equation}
\vec{F}(\vec{\Psi})=
-P_{\mathcal{D}}^{\perp}
\left(
\dfrac{\delta E}{\delta\vec{{\Psi}}^{\dagger}}
\right)-
\sum_{j=n_P+1}^{n_c}\lambda_{j}
P_{\mathcal{D}}^{\perp}\left(
\dfrac{\delta G_{j}}{\delta\vec{{\Psi}}^{\dagger}}
\right).
\label{eq:num_test_F}
\end{equation}
We assume that we have $N$ discretization points in space and store the numerical values of 
$\vec{\Psi},\vec{\Phi},\vec{F}$ at iteration $k$ as matrices $\mathbf{U}, \mathbf{V},  \mathbf{F}\in \mathbb{R}^{N\times n}$, respectively.
Then for some initial conditions $\mathbf{U}_0,\mathbf{V}_0$  we can apply the symplectic Euler algorithm, see \cite{HairerLubichWanner2006}, to~(\ref{eq:num_test_sys}) to get the following semi-implicit method
\begin{equation}
\begin{aligned}
&\mathbf{V}_{k+1} = (I - \eta \, \Delta \tau)\mathbf{V}_k + \Delta \tau\,  \mathbf{F}_k, \\
&\mathbf{U}_{k+1} = \mathbf{U}_k + \Delta \tau\,  \mathbf{V}_{k+1} \:,\, k = 0,1,... ,
\end{aligned}
\end{equation}
where $\Delta \tau$ is the  step in artificial time. Note that the symplectic Euler above has a dual version where $\mathbf{U}_{k}$ is updated before $\mathbf{V}_{k}$ but with the same numerical properties (convergence and efficiency) \cite{HairerLubichWanner2006}.

\begin{figure}[h]
\includegraphics[scale=0.26]{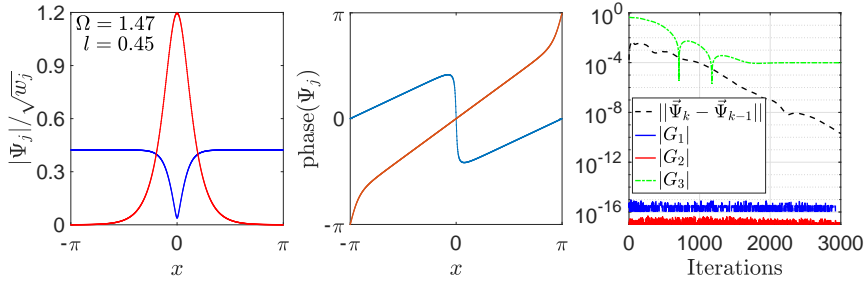}
\caption{(Color online) Stationary solution of~(\ref{eq:Num_1D_2_comp}),~(\ref{eq:lambda_3_for_1D_2-comp}) with projection for the normalization constants, see Sec.~\ref{sec:proj_all_norm}.
Modulus (left) and phases (mid), of $\Psi_j(x), \: j=1,2$. Compare with FIG. 4 (most right) of~\cite{WuZaremba2013}. Convergence (right) for $\vec{\Psi}$ and the constraint $G_3$ of (\ref{eq:G1G2_only_angularmomentum}), see inset legend. We have used the numerical parameters $\eta=k_3=2$ and $\Delta \tau=0.99\Delta x \simeq 6.2 \cdot 10^{-3}$. Physical parameters are set to $\gamma=\gamma_{11}=\gamma_{12}=\gamma_{22}=23$, $c_1=0.96$, $c_2=0.04$, and $c_3=0.45$, which is the same as in \cite{WuZaremba2013}.} \label{fig:first_1D}
\end{figure}

\subsection{Multicomponent NLSE in 1D under rotation in a ring, applications to vector solitons}

In these examples we want to find normalized solutions $\vec{\Psi} (x)$ of (\ref{eq:Euler_Lagrange_to_energy_functional2}) where $W = \text{diag}(1, \ldots, 1)$ and $V = \text{diag}(0, \ldots, 0)$, i.e.,
\begin{equation}
\frac{\delta E(\vec{{\Psi}})}{\delta\vec{{\Psi}}^{\dagger}}=
-\frac{\partial^{2}\vec{{\Psi}}}{\partial x^{2}} +2\pi\Gamma\left(\vec{{\Psi}}\right)\vec{{\Psi}} = 0,
\label{eq:Euler_Lagrange_to_energy_functional2_W_1_V_0}
\end{equation}
with a constrained angular momentum.

\subsubsection{Projection of all the normalization constraints}
\label{sec:proj_all_norm}

From Eqs.~(\ref{eq:App_A_constraints_on_component_form}),~(\ref{eq:Kdef}), and~(\ref{eq:definition_angular_momentum_operator}) we are left with one dynamic constraint for the angular momentum
\begin{equation}
\label{eq:G1G2_only_angularmomentum}
G_{n+1}  =c_{n+1} + i \sum_{j=1}^n\int\limits_{-\pi}^{\pi} \overline{\Psi}_j  \frac{\partial}{\partial x}\Psi_j
\, {\rm d}x  =0.
\end{equation}
We obtain the Lagrange parameter $\lambda_{n+1}$ in each time step for the angular momentum from~(\ref{eq:lambda_ell_for_projected_norms}),
with $L$ from (\ref{eq:definition_angular_momentum_operator}) and $\delta E(\vec{{\Psi}}) / \delta\vec{{\Psi}}^{\dagger}$ from (\ref{eq:Euler_Lagrange_to_energy_functional2_W_1_V_0}).
The normalization of each component, $\int \left| \Psi_j \right|^2 dx=c_j$, is here fulfilled by (trivial) projection,
$\Psi_j = \sqrt{c_j} \Psi_j/\| \Psi_j \|$,
in each time step.

\subsubsection{Only dynamically damped constraints}
\label{sec:dyn_damp_constr}
From Eqs.~(\ref{eq:App_A_constraints_on_component_form}),~(\ref{eq:Kdef}), and~(\ref{eq:definition_angular_momentum_operator}) we end up with $n+1$ dynamic constraints 
\begin{equation}
\label{eq:G1G2}
G_{j}=c_{j}-\int\limits_{-\pi}^{\pi} \left| \Psi_j \right|^2
\, {\rm d}x  =0,\ j=1,\: ..., n, \
G_{n+1}  =c_{n+1} + i \sum_{j=1}^n\int\limits_{-\pi}^{\pi} \overline{\Psi}_j  \frac{\partial}{\partial x}\Psi_j
\, {\rm d}x  =0.
\end{equation}
We obtain the corresponding Lagrange parameters $\vec{\lambda}$ from the linear system given by the matrix $A$ in~(\ref{eq:AllDynA}) and the vector $\vec{b}$ in~(\ref{eq:AllDynb}) by solving $A\vec{\lambda} = \vec{b}$ in each time step.

\subsubsection{Numerical benchmarking for $n=2$ components}
In the special case of two components there is a rich literature on applications in for example optics and cold atomic gases, see, e.g.,~\cite{WuZaremba2013,RoussouEtAl2018, SandinOgrenGulliksson2016} and references therein.

For the ease of the readers familiar with those applications,
we here explicitly write the corresponding dynamic equations~(\ref{eq:num_test_sys}) and~(\ref{eq:num_test_F}) for the stationary equation (\ref{eq:Euler_Lagrange_to_energy_functional2_W_1_V_0}) in the case of two components. 
With $\vec{\Psi} = \left[ \Psi_1, \Psi_2 \right]^T$ and $\vec{\Phi} = \left[ \partial \Psi_1 / \partial \tau, \partial \Psi_2 / \partial \tau\right]^T$, we have
$$
\frac{\partial^{2}}{\partial x^{2}}
\left[
\begin{array}{c}
\Psi_1 \\
\Psi_2 
\end{array}
\right]
 -
2 \pi \left[
\begin{array}{cc}
\gamma_{11} \left|
 \Psi_1
\right|^2 & \gamma_{12} \Psi_1 \overline{\Psi}_2 \\
\gamma_{21} \overline{\Psi}_1 \Psi_2& \gamma_{22} \left|
 \Psi_2
\right|^2 
\end{array}
\right]
\left[
\begin{array}{c}
\Psi_1 \\
\Psi_2 
\end{array}
\right]
$$
\begin{equation}
    \label{eq:Num_1D_2_comp}
+ \lambda_1
\left[
\begin{array}{c}
\Psi_1 \\
0 
\end{array}
\right]
+ \lambda_2
\left[
\begin{array}{c}
0\\
\Psi_2 
\end{array}
\right]
+ \lambda_3 L
\left[
\begin{array}{c}
 \Psi_1 \\
 \Psi_2 
\end{array}
\right]=0,
\end{equation}
where $\lambda_1=\lambda_2=0$, and~(\ref{eq:lambda_ell_for_projected_norms})
\begin{equation}
\lambda_3= 
\dfrac{
k_{3}G_{3}/2 + \sum_{l=1}^{2} 
 \left\langle
\dfrac{\delta E}{\delta\overline{\Psi}_l}   \middle|   L \Psi_l
\right\rangle
-
 \left\langle
 \Psi_l   
\middle| L\Psi_l 
 \right\rangle 
 \left\langle
 \dfrac{\delta E}{\delta\overline{\Psi}_l}  
  \middle|
  \Psi_l 
\right\rangle  /c_l
-
 \left\langle  
\Phi_l \middle| L\Phi_l 
\right\rangle
}
{
\sum_{l=1}^{2} 
 \left\langle
\Psi_l  \middle|   L^2 \Psi_l
\right\rangle
- 
\left\langle  \Psi_l \middle| L \Psi_l
\right\rangle^2/c_l
},
\label{eq:lambda_3_for_1D_2-comp}
\end{equation}
for the case of projecting the two normalization constraints.
The case of only dynamically damped constraints, the three Lagrange parameters are in each time step given from the linear system~(\ref{eq:AllDynA}),~(\ref{eq:AllDynb})
$$
\left[
\begin{array}{ccc}
 \left\langle {\Psi}_{1} \middle| {\Psi}_{1} \right\rangle  &  0 & \left\langle {\Psi}_{1} \middle| L{\Psi}_{1} \right\rangle  \\
0 &   \left\langle {\Psi}_{2} \middle| {\Psi}_{2} \right\rangle  & \left\langle {\Psi}_{2} \middle| L{\Psi}_{2} \right\rangle  \\
 \left\langle {\Psi}_{1} \middle| L{\Psi}_{1} \right\rangle 
&   \left\langle
 {\Psi}_{2} \middle| L{\Psi}_{2} 
 \right\rangle 
& \left\langle {\Psi}_{1} \middle| L^2{\Psi}_{1} \right\rangle  + \left\langle {\Psi}_{2} \middle| L^2{\Psi}_{2} \right\rangle
\end{array}
\right]
\left[
\begin{array}{c}
\lambda_1 \\
\lambda_2 \\
\lambda_3 
\end{array}
\right]
$$
\begin{equation}
=
\left[
\begin{array}{c}
k_1 G_1 /2 +
\left\langle 
\Psi_{1}\middle|\dfrac{\delta E}{\delta\overline{\Psi}_1}\right\rangle -
\|  \Phi_1 \| ^2  \\
k_2 G_2 /2 +
\left\langle 
\Psi_{2}\middle|\dfrac{\delta E}{\delta\overline{\Psi}_2}\right\rangle -
\|  \Phi_2 \| ^2  \\
k_3 G_3 /2 +
\left\langle \Psi_{1}\middle|\dfrac{\delta E}{\delta\overline{\Psi}_1}\right\rangle + \left\langle \Psi_{2}\middle|\dfrac{\delta E}{\delta\overline{\Psi}_2}\right\rangle -
\left\langle  \Phi_1 \middle| L \Phi_1 \right\rangle  + \left\langle  \Phi_2 \middle| L \Phi_2 \right\rangle 
\end{array}
\right].
\label{eq:3_3_system_for_1D_2-comp}
\end{equation}
Above $L$ is the angular momentum operator~(\ref{eq:definition_angular_momentum_operator}), $G_1$, $G_2$, $G_3$ are taken from~(\ref{eq:G1G2}), and the functional derivatives of the energy are given by~(\ref{eq:Euler_Lagrange_to_energy_functional2_W_1_V_0}).

To compare our numerical results in 1D with the literature, we first reproduced some well-known solutions with~(\ref{eq:Num_1D_2_comp}), which can be given analytically for certain parameter values~\cite{WuZaremba2013}.
We perform the calculations with both the dynamic methods presented here for obtaining the Lagrange parameters.
We present the numerical results for the case
of projecting all the normalization constraints, described 
in Sec.~\ref{sec:proj_all_norm}, in Fig.~\ref{fig:first_1D},
and for the case
of all the constraints being dynamic, described 
in Sec.~\ref{sec:dyn_damp_constr}, in Fig.~\ref{fig:first_1D_2}.

Note that even if the physical quantities turns out the same, left and mid plots of Figs.~\ref{fig:first_1D} and~\ref{fig:first_1D_2}, the convergence for the constraints seen in the right plots are qualitatively different.
In order to reach even higher accuracy for the angular momentum constraint $G_3$, one need a smaller $\Delta x$.

\begin{figure}[h]
\includegraphics[scale=0.26]{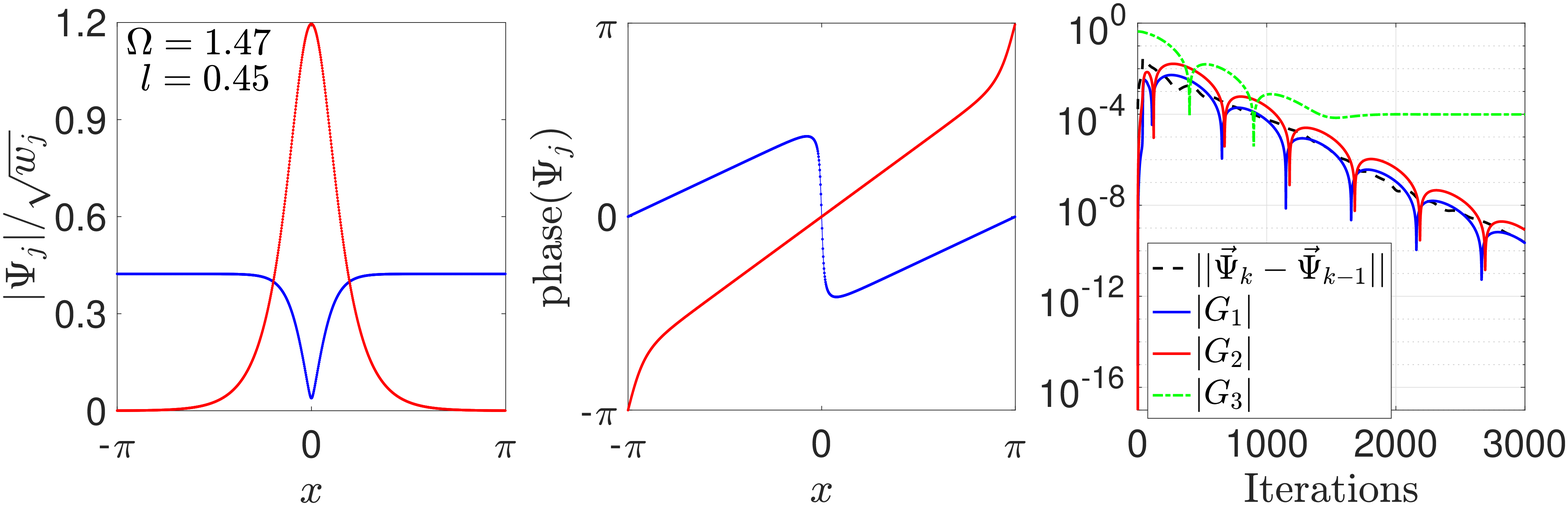}\caption{(Color online) Stationary solution of~(\ref{eq:Num_1D_2_comp}),~(\ref{eq:3_3_system_for_1D_2-comp}) with fully dynamic constraints, see Sec.~\ref{sec:dyn_damp_constr}. Everything is the same as for Fig.~\ref{fig:first_1D}, except that  we here have $3$ damping constants for the $3$ dynamic constraints~(\ref{eq:G1G2}), i.e., $\eta=k_1=k_2=k_3=2$, and a qualitatively different convergence of the constraints, $G_1, G_2$, see the right plot.}
\label{fig:first_1D_2}
\end{figure}

\subsubsection{Computational complexity for $n \leq 10$ components}

Theoretical and experimental studies of physical systems described by the NLSE with more than two components is a rapidly growing field, see, e.g.,~\cite{GeneralNsolitonFeng2014, LannigEtAl2020} and references therein.
In this light, efficient numerical methods to treat such systems for any physical parameters, as we have formulated in this article, are timely.
In this section we evaluate numerically the computational complexity of the presented dynamic methods with up to ten components. 

If one or several components represent a very small part of the total density, they will not  influence the other components much, which may result in numerical results not showing general behaviour of our method but rather numerical artifacts.
Hence, given that we want to test up to $n \leq 10$ components numerically here, we set a lower bound for each density to $1/(2n)$, such that no component will contain less than $5 \%$ of the total density.

Another possible qualitative problem when evaluating different number of components with similar parameter values that we want to decrease, is that many aspects of the solutions are dictated by number theoretical effects, as have been discussed for $n=2$ components in~\cite{RoussouEtAl2018}.
Therefore we add to component $k$ a term $a^k$ from a geometric series with $n$ terms and sum $1-n \cdot 1/(2n)  = 1/2$, such that
$\sum_{k=1}^{n} c_k = \sum_{k=1}^{n} 1/(2n) + a^k =1$, i.e., $0 < a <1$ is a (real) solution to $a^{n+1} + 3a/2 + 1/2 = 0$.
Then we obtain for example: $c_1=1$ for $n=1$; $c_1=0.6160, \: c_2=0.3840$ for $n=2$; and $ c_1=0.5092, \: c_2=0.2840, \: c_3=0.2068$ for $n=3$.

The strength of the parameters for the nonlinear coupling were here set to $\gamma_{ii}=10$ and  $\gamma_{ij}=5, \: i \neq j$ (which do not give spatial separations for the different components). 
 
Initial conditions used for the calculations in all Figs.~\ref{fig:first_1D},~\ref{fig:first_1D_2} and~\ref{second_1D_figure} are $\Psi_j(x,0)=$ $\sqrt{c_j/(4\pi)}(1 + \exp(i x))$, for each component $j=1, 2, \ …, n$, such that $\langle \vec{\Psi} \mid \vec{\Psi} \rangle=1$ and $\langle \vec{\Psi} \mid L \vec{\Psi} \rangle=0.5$ for $\tau=0$.
As is natural, we observed that the convergence is faster for calculations with $\ell_0$ close to the initial value $\ell = 0.5$, but the number of iterations that we use to calculate the computational complexity is the average of the $9$ different values $\ell_0=0.1, \ 0.2, \ …, 0.9$. 
For one case of the data reported in the green curve of Fig.~\ref{second_1D_figure}, we did not reach any convergence at all for the above described choices, which also represents a specific indication of our general obtained experience that the \textit{SHAKE and RATTLE} \cite{AndersenRattle1983,McLachlanModinVerdierWilkins2014} method is further more sensitive to the initial conditions than the methods with dynamic constraints. 
Instead, for this case we both reduced $\Delta \tau=0.9 \Delta x$ and used $\Psi_2(x,0)=\sqrt{c_2/(4\pi)} (1 + \exp(1.1 i x))$ in order to reach convergence.

In addition to the obtained (averaged over $c_{n+1}=\ell_0$) number of iterations $k$ needed for the norm $ || \vec{\Psi}_k - \vec{\Psi}_{k-1}  || < 10^{-6}$, we multiplied with an overall factor $n$ describing the complexity in setting up the equations to obtain the Lagrange parameters of each method, see Eqs.~(\ref{eq:lambda_ell_for_projected_norms}), and~(\ref{eq:AllDynA}),~(\ref{eq:AllDynb}), respectively. 
For the \textit{SHAKE and RATTLE} we use an additional factor $4$ coming from the additional (half-) step in updating the numerical solution together with $3$ Newton iterations, which was the observed minimum number of iterations in the inner solver for the Lagrange parameters.  
The time step in use for the benchmarking in Fig.~\ref{second_1D_figure},  is the seemingly maximal possible  $ \Delta \tau < \Delta x $ (we set $\Delta \tau=0.99 \Delta x$) However, we have not been able to justify this theoretically.
%\cite{DiscreteLaplace}.
The damping parameters are in all cases set to $\eta=k_j=2$, which gives good performance, but we stress that we did not perform any individual optimization of the damping parameters for each case. 
For the \textit{SHAKE and RATTLE} method \cite{AndersenRattle1983,McLachlanModinVerdierWilkins2014} one needs in addition initial conditions for the Lagrange parameters, which we set to $\lambda_j=1$ with $j=1,2$ for $n=1$, and $j=1,2,3$ for $n=2$.
%!!! KAN VI SAGA NAGOT OM ATT ANTALET ITERATIONER (GENERELLT) SKALAR SOM n ??? Kanske...
The substantially larger errorbars for some of the $n$-values in Fig.~\ref{second_1D_figure}, is due  to number theoretical effects~\cite{RoussouEtAl2018}, even if we have chosen parameters as to reduce such effects.
\begin{figure}[h]
\begin{center}
\includegraphics[scale=0.4]{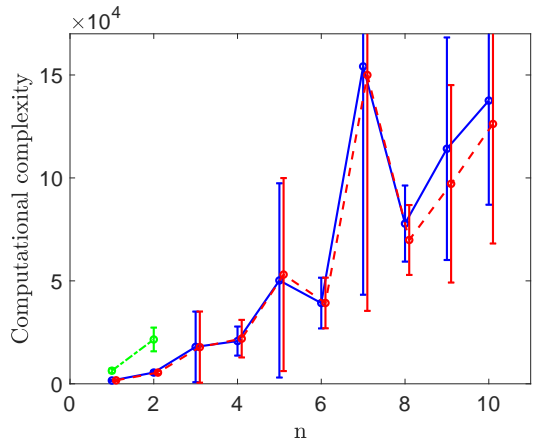}
\end{center}
\caption{(Color online) Benchmarking of the computational complexity for constrained $n$-component solutions of (\ref{eq:Euler_Lagrange_to_energy_functional2_W_1_V_0}).
The solid (blue)  data is for \textit{Only dynamically damped constraints} (Sec. \ref{sec:dyn_damp_constr}), dashed (red) data (slightly translated to the right of the integer values $n$ for improved visibility) is for \textit{Projection of all the normalization constraints} (Sec.~\ref{sec:proj_all_norm}), dashed-dotted (green)  data is for \textit{SHAKE and RATTLE} \cite{McLachlanModinVerdierWilkins2014} with available algorithms for $n=1$ and $n=2$ taken from \cite{SandinOgrenGulliksson2016}.  
The computational complexity for each method has been calculated by averaging over different $\ell$ values as described in the main text, and the mean (curves) and standard deviation (half the total length of the errorbars) are formed from the number of iterations until convergence, multiplied with different factors describing the three different algorithms under study here, see main text.}
\label{second_1D_figure}
\end{figure}

\subsection{NLSE in 2D under rotation in a harmonic oscillator potential, applications to vortex structures}

Theoretical and experimental studies of quantized vortices in, e.g., liquid Helium, superconductors, and cold atomic gases, have been  undertaken for many years.
Solving the NLSE under rotation is a common but non-trivial such example, see, e.g.,~\cite{AntoineDuboscg2014,JacksonPRA2004, LinnPRA2001, KavoulakisMottelsonPethick2000, ButtsRokhsar1999, BaoMarkowich2005} and references therein.

In the 2D examples presented here we want to find  $\Psi (x_1,x_2)$ where
\begin{equation}
\label{eq:oscEnergy}
E\left(\Psi\right)=
\frac{1}{2}\int\limits_{\mathbb{R}^2}
\left|
\dfrac{\partial \Psi}{\partial x_1}
\right|^2+
\left|
\dfrac{\partial \Psi}{\partial x_2}
\right|^2+
 \left( x_1^{2} + x_2^{2} \right)\left| \Psi \right|^2+
 g\left|
 \Psi
\right|^4 \, {\rm d}x_1 {\rm d}x_2 ,
\end{equation}
gives the stationary equation
\begin{equation}
    \label{eq:NumRot}
\frac{\delta E(\Psi)}{\delta \Psi^{\dagger}}=
-\dfrac{1}{2}\left(\frac{\partial^{2}\Psi}{\partial x_1^{2}} +\frac{\partial^{2}\Psi}{\partial x_2^{2}}\right)+
\dfrac{1}{2}\left( x_1^{2} + x_2^{2} \right)\Psi+
g\left|
 \Psi
\right|^2\Psi =0.
\end{equation}
The dynamic constraints are 
\begin{equation}
\label{eq:G1G2_2D}
G_{1}=c_{1}-\int\limits_{\mathbb{R}^2} \left| \Psi \right|^2
\, {\rm d}x_1 {\rm d}x_2 =0,\
G_{2}  =c_{2}-\int\limits_{\mathbb{R}^2} \Psi^{\dagger} L \Psi
\, {\rm d}x_1 {\rm d}x_2 =0,
\end{equation}
with the angular momentum operator
\begin{equation}
L = -i(x_1\frac{\partial }{\partial x_2} -x_2\frac{\partial }{\partial x_1}) .
\end{equation}
We can generally attain the corresponding Lagrange parameters $\vec{\lambda}$ from the linear system given by the matrix $A$ in~(\ref{eq:AllDynA}) and the vector $b$ in~(\ref{eq:AllDynb}) by solving $A\vec{\lambda} = b$.
\vspace{40mm}

\begin{figure}[h]%[htb]
\begin{center}
\includegraphics[scale=0.75]{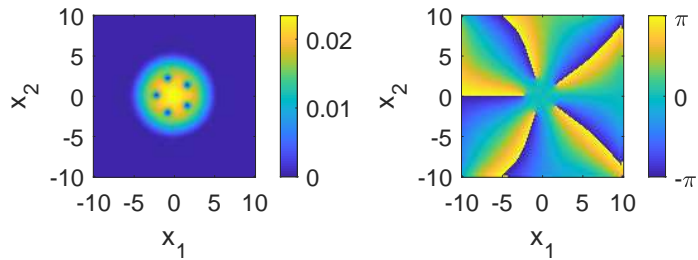}
\end{center}
\caption{(Color online) Stationary solution of (\ref{eq:DFPM_2D_NLSE_with_rotation}),
density (left), and phase (right), of $\Psi (x_1, x_2)$.  
Compare with Figure 1 of \cite{AntoineDuboscg2014}.
We used the damping $\eta=4$ and $\Delta \tau=0.1$.
\label{fig:first_2D_figure}}
\end{figure}

\subsubsection{The groundstate of a rapidly rotating Bose-Einstein condensate}

To compare our numerical results in 2D with the literature, we first reproduced \textit{"A simple but complete example"} from 
the dimensionless GPE in Sec.~6 of~\cite{AntoineDuboscg2014}. 
In this example the rotational velocity
$\Omega$ is kept constant, so we used the following dynamical reformulation
of (\ref{eq:NumRot})
\begin{equation}
\frac{\partial^{2}\Psi}{\partial\tau^{2}}+\eta\frac{\partial\Psi}{\partial\tau}=\frac{1}{2}\left(\frac{\partial^{2}\Psi}{\partial x_{1}^{2}}+\frac{\partial^{2}\Psi}{\partial x_{2}^{2}}\right)-\frac{1}{2}\left(x_{1}^{2}+x_{2}^{2}\right)\Psi-500\left|\Psi\right|^{2}\Psi+\lambda_{1}\Psi+\lambda_{2}L\Psi,\label{eq:DFPM_2D_NLSE_with_rotation}
\end{equation}
where $\lambda_{2}=\Omega$ (constant), $\lambda_{1}=\mu=E_{\textnormal{{rot}}}+\int\int 250\left|\Psi\right|^{4} + \Psi^\dagger \left( \frac{\partial^{2}\Psi}{\partial\tau^{2}}+\eta\frac{\partial\Psi}{\partial\tau} \right) dx_{1}dx_{2}$,
and with projection $\Psi=\Psi/\left\Vert \Psi\right\Vert $ in each
time step to fulfill the normalization constraint. 
Here $E_{\textnormal{{rot}}}$
is the energy in the rotating frame

\begin{equation}
E_{\textnormal{{rot}}}\left(\Omega\right)=\int_{\mathbb{R}^2}\frac{1}{2}\left|\nabla\Psi\right|^{2}+\frac{1}{2}\left(x_{1}^{2}+x_{2}^{2}\right)\left|\Psi\right|^{2}+250\left|\Psi\right|^{4}-\Omega\Psi^{*}L\Psi \, dx_{1}dx_{2}.\label{eq:2D_energy_in_the_rotating_frame}
\end{equation}
As in sec. 6 of \cite{AntoineDuboscg2014} we use a spatial grid $-10\leq x_{1},x_{2}\leq10$
with $129$ points in each dimension, and  the initial
condition {[}Eq. (3.23) of \cite{AntoineDuboscg2014}, with $\gamma_{x}=\gamma_{y}=1${]},

\begin{equation}
\Psi\left(x_{1},x_{2},0\right)=\frac{\left[\left(1-\Omega\right)+\Omega\left(x_{1}+ix_{2}\right)\right]\exp\left(-x_{1}^{2}/2-x_{2}^{2}/2\right)}{\left\Vert \left[\left(1-\Omega\right)+\Omega\left(x_{1}+ix_{2}\right)\right]\exp\left(-x_{1}^{2}/2-x_{2}^{2}/2\right)\right\Vert }.\label{eq:2D_initial_condition_Bao}
\end{equation}

Performing the numerical calculation solving (\ref{eq:DFPM_2D_NLSE_with_rotation}), we obtain for $\Omega=0.5$
a stationary solution with a density $\rho=\left|\Psi\right|^{2}$,
which is 5-fold symmetric, depicted in Fig. \ref{fig:first_2D_figure}.

\begin{figure}[h!]
\begin{center}
\includegraphics[scale=0.75]{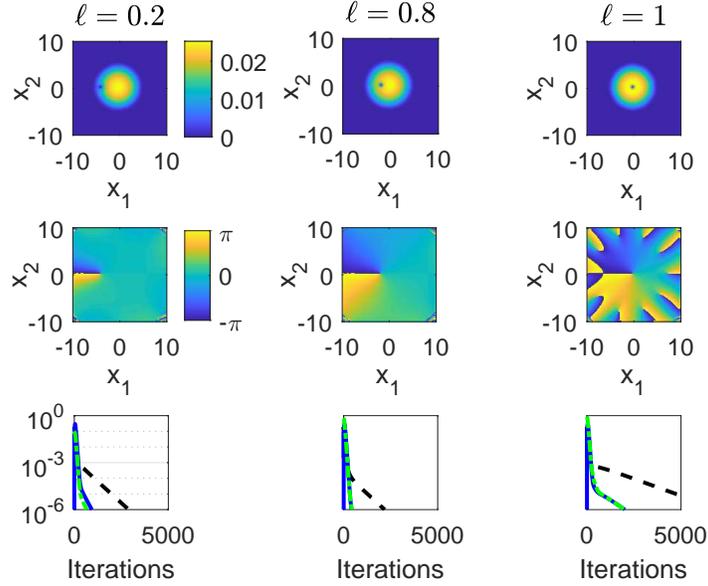}
\end{center}
\caption{(Color online) Nucleation of a central vortex studied with dynamical constraints.
The upper panels shows the density $\rho=\left|\Psi (x_1,x_2) \right|^{2}$ for
angular momentum $c_2 = \ell=0.2,\:0.8,\:1$, and the middle panels show
the phase of $\Psi$, for stationary solutions of (\ref{eq:DFPM_2D_NLSE_with_rotation})
with $\mu$ and $\Omega$ being Lagrange parameters. In the lower
panels the corresponding numerical convergence of the DFPM and the
dynamical constraints are illustrated. The (black) dashed curve shows
$\left\Vert \Psi_{k}-\Psi_{k-1}\right\Vert $ for iteration step $k$,
the (blue) solid curve shows the normalization constraint $\left|G_{1}\left(\Psi_{k}\right)\right|$,
and the (green) dotted-dashed curve shows the angular momentum constraint
$\left|G_{2}\left(\Psi_{k}\right)\right|$. For this figure we used
the numerical parameters $\eta=k_{1}=k_{2}=2$ and $ \Delta \tau=0.05$.
\label{fig:second_2D_figure}}
\end{figure}

\subsubsection{Transition of a non rotating Bose-Einstein condensate forming a central
vortex}
After the first test of DFPM without dynamical constraints above,
we now test the method with two dynamical constraints on a well-known
example, the transition of a distant vortex approaching the center
$\left(x_{1},x_{2}\right)=\left(0,0\right)$ for $0<\ell\leq1$. Compare our
Fig. \ref{fig:second_2D_figure}, for example with Figure 3 in \cite{ButtsRokhsar1999},
or Fig. 3 in \cite{KavoulakisMottelsonPethick2000}. This means we are now using
Eq. (\ref{eq:DFPM_2D_NLSE_with_rotation}) with the two Lagrange multipliers
evaluated dynamically, corresponding to Eqs.~(\ref{eq:B_10_in_EJP_2020_upper}),~(\ref{eq:B_10_in_EJP_2020_lower}), and the same initial condition, Eq. (\ref{eq:2D_initial_condition_Bao})
with $\Omega=0.1$, for all the three different constrained $\ell$-values
in Fig.~\ref{fig:second_2D_figure}. The densities (upper plots) and
phases (middle) of the stationary solution to (\ref{eq:DFPM_2D_NLSE_with_rotation})
behaves as expected, and the convergence of the iteration and the
constraints (\ref{eq:G1G2_2D}) are shown in the lower plots.

In conclusion we find the presented method with dynamic constraints an effective tool to calculate
vortex structures in 2D with no assumptions on the strength of the
nonlinearity or the form of the external potential. 
For example we
have also successfully tested to calculate giant vortices for an additional quartic term
in the potential~\cite{JacksonPRA2004}, and rotating states for an asymmetric potential~\cite{LinnPRA2001}.
We have
kept the damping parameters and the time step constant throughout
this subsection to illustrate the robustness of the method. However,
we noted that the numerical performence can be substantially improved
by finetuning those numerical parameters.

\section{Conclusions}
\label{sec:Conclusions}

We have set up a general dynamical formulation  Eqs.~(\ref{eq:DFPMproj}),~(\ref{eq:Gdynm}),~(\ref{eq:genA}),~(\ref{eq:genb}) for obtaining stationary solutions for nonlinear equations with $n_c$ global constraints, where $0 \leq n_P \leq n_c$ of them are projected and $n_c - n_P$ of them are included in the (extended) dynamical formulation.

We presented detailed derivations for the case of rotation in 1D system, with examples of  treating the constraints only with projection, Sec.~\ref{sec:OnlyProj}, projection of some, Sec.~\ref{subsec_proj_1_or_more_norm}, and all, Sec.~\ref{sec:ProjAllNorm}, normalization constraints, and finally with only dynamic constraints, Sec.~\ref{sec:OnlyDamp}.

In the final part of the article we have given detailed formulae for some 1D, Figs.~\ref{fig:first_1D}, and~\ref{fig:first_1D_2}, and 2D, Figs.~\ref{fig:first_2D_figure}, and~\ref{fig:second_2D_figure}, examples that we have evaluated numerically and compared to the literature.
We also evaluated the computational performance numerically with up to ten components, Fig.~\ref{second_1D_figure}.

In conclusion, the dynamical formulation of complicated global constraints have the following main advantages:

\begin{itemize}
\item Simpler implementations
\item Less computational complexity %[Fig.~\ref{second_1D_figure}]
\item Less sensitivity to initial conditions
\end{itemize}

We do not report any disadvantages with the methods studied. 
It may appear as that we need to chose additional suitable numerical parameter values (the $k_j$:s in (\ref{eq:Gdynm})) for the damping of the constraints.
However, with projection methods one needs instead to chose suitable initial values for the Lagrange parameters.   

\bibliographystyle{unsrt}

\begin{thebibliography}{99}

\bibitem{Fibich2015} Fibich G 2015, “The Nonlinear Schrödinger Equation”, Springer, New York. ISBN 978-3-319-12747-7. 

\bibitem{OgrenEtAlSolitons2017} Ögren M, Abdullaev F Kh, and Konotop V V 2017, ”Solitons in a PT-symmetric $\chi^{(2)}$ coupler”,
Opt. Lett. \textbf{42}, 4079. 
 
\bibitem{SorensenEtAl2017} S\o rensen M P, Pedersen N F, and Ögren M 2017, “The dynamics of magnetic vortices in type II superconductors with pinning sites studied by the time dependent Ginzburg-Landau model”, 
Physica C \textbf{533}, 40. 
 
\bibitem{ChavanisEtAl2011} Chavanis P-H 2011, “Mass-radius relation of Newtonian self-gravitating Bose-Einstein condensates with short-range interactions. I. Analytical results”, Phys. Rev. D \textbf{84}, 043531. 
 
\bibitem{RozemannEtAl2020} Rozenman G G, Shemer L, and Arie A 2020, “Observations of accelerating solitary wavepackets”, Phys. Rev. E \textbf{101}, 050201(R). 

\bibitem{SystheEtAl2008} Dysthe K, Krogstad H E, and M\"{u}ller P 2008, “Oceanic Rogue Waves”, Annu. Rev. Fluid Mech. \textbf{40}, 287. 

\bibitem{GeneralNsolitonFeng2014}
Feng B, F, 2014, "General N-soliton solution to a vector nonlinear Schrödinger equation",
J. Phys. A: Math. Theor.
{\bf 47}, 355203. 

\bibitem{BeattieEtAl2013} Beattie S, Moulder S, Fletcher R J, and Hadzibabic Z 2013, “Persistent Currents in Spinor Condensates”, Phys. Rev. Lett. \textbf{110}, 025301. 

% 3-components:
\bibitem{LannigEtAl2020} Lannig S, Schmied C M, Prüfer M, Kunkel P, Strohmaier R, Strobel H, Gasenzer T, Kevrekidis P G, and Oberthaler M K, 2020,
"Collisions of Three-Component Vector Solitons in Bose-Einstein Condensates", 
Phys. Rev. Lett. {\bf 125}, 170401. 
 
\bibitem{OgrenEtAlSolitary2005} Ögren M, Kavoulakis G M, and Jackson A D 2005, “Solitary waves in elongated clouds of strongly-interacting bosons”,
Phys. Rev. A \textbf{72}, 021603(R). 
 
\bibitem{OgrenEtAlRotational2020} Ögren M and Kavoulakis G M 2020, “Rotational properties of superfluid Fermi-Bose mixtures in a tight toroidal trap”,
Phys. Rev. A .\textbf{102}, 013323.

\bibitem{AntoineDuboscg2014}
Antoine X and Duboscq R, 2014,
"GPELab, a Matlab Toolbox to Solve Gross-Pitaevskii Equations I: Computation
of Stationary Solutions, Computer Physics Communications" \textbf{185}
(11),  pp. 2969-2991.

\bibitem{SmyrlisEtAl2004} Smyrlis G and Zisis V 2004, “Local convergence of the steepest descent method in Hilbert spaces”, Journal of Mathematical Analysis and Applications
, \textbf{300}, 436.

\bibitem{NoceWrig06_Numerical_Optimization} Nocedal J and  Wright S J 2006, ``Numerical Optimization'', 2nd ed. Springer, ISBN-13 978-0387-30303-1. 

\bibitem{SraEtAl2012} Sra S, Nowozin S, and Wright S J 2012, "Optimization from Machine Learning", MIT Press, ISBN 9780262016469

\bibitem{ZeidlerIII} 
Boron L F and Zeidler E 1984,
''Nonlinear Functional Analysis and its Applications: III: Variational Methods and Optimization'',
%  isbn={9780387909158},
%  lccn={83020455},
%url={https://books.google.se/books?id=JcfmzAEACAAJ},
Springer New York.%%, 1984.

\bibitem{Gulliksson_book_chapter_2019} Gulliksson M, Ögren M, Oleynik A, and Zhang Y 2018, ``Damped Dynamical Systems for Solving Equations and Optimization Problems''. In: Sriraman B. (eds) Handbook of the Mathematics of the Arts and Sciences. Springer, Cham, ISBN 978-3-319-70658-0.

\bibitem{OgrenGulliksson2020}
Ögren M and Gulliksson M 2020, 
"A numerical damped oscillator approach to constrained Schrödinger equations",
Eur. J. Phys. {\bf 41}, 065406. 

\bibitem{HairerLubichWanner2006} Hairer E, Lubich C, and Wanner G 2006, ``Geometric Numerical Integration'', 2nd ed. Springer, ISBN 978-3-540-30666-5. 

\bibitem{WuZaremba2013}
Wu Z and Zaremba E 2013,
"Mean-field yrast spectrum of a two-component Bose gas in ring geometry: 
Persistent currents at higher angular momentum", 
Phys. Rev. A {\bf 88}, 063640.

\bibitem{RoussouEtAl2018}
Roussou A, Smyrnakis J, Magiropoulos M, Efremidis N K, Kavoulakis G M, Sandin P, Ögren M, and Gulliksson M, 2018,
"Excitation spectrum of a mixture of two Bose gases confined in a ring potential with interaction asymmetry.",
New J. Phys. {\bf 20}, 045006.

\bibitem{SandinOgrenGulliksson2016} Sandin P, Ögren M, and Gulliksson M 2016, ``Numerical solution of the stationary multicomponent nonlinear Schr\"odinger equation with a constraint on the angular momentum'', Phys. Rev. E {\bf 93}, 033301.

\bibitem{AndersenRattle1983}
Andersen H C 1983,
"Rattle: A 'velocity' version of the shake algorithm for molecular
dynamics calculations", Journal of Computational Physics, \textbf{52} (1):24 – 34.

\bibitem{McLachlanModinVerdierWilkins2014} McLachlan R, Modin K, Verdier O, and  Wilkins M 2014, ``Geometric generalisations of SHAKE and RATTLE'',
Foundations of Computational Mathematics \textbf{14}, 339.

\bibitem{JacksonPRA2004}
Jackson A D, Kavoulakis G M, and Lundh E, 2004,
"Phase diagram of a rotating Bose-Einstein condensate with anharmonic confinement", Phys. Rev. A \textbf{69}, 053619.

\bibitem{LinnPRA2001}
Linn M, Niemeyer M, and Fetter A L, 2001,
"Vortex stabilization in a small rotating asymmetric Bose-Einstein condensate", Phys. Rev. A \textbf{64}, 023602.

\bibitem{KavoulakisMottelsonPethick2000}
Kavoulakis G M, Mottelson B, and Pethick C J, 2000,
"Weakly interacting Bose-Einstein condensates under rotation", Phys.
Rev. A \textbf{62}, 063605.

\bibitem{ButtsRokhsar1999}
Butts D A and Rokhsar D S, 1999,
"Predicting signatures of rotating Bose-Einstein condensates",
Nature \textbf{397}, 327.

\bibitem{BaoMarkowich2005}
Bao W, Wang H, and Markowich P A, 2005,
"Ground, symmetric and central
vortex states in rotating Bose-Einstein condensates", Communications
in Mathematical Sciences, 3(1):57-88.

\end{thebibliography}

\end{document}